\documentclass[12pt]{article}%
\usepackage{authblk}
\usepackage{amssymb}
\usepackage{IEEEtrantools}
\usepackage{amsfonts}
\usepackage{amsmath}
\usepackage{array}
\usepackage{enumerate}
\usepackage{mathtools}
\usepackage{cite}
\usepackage[nohead]{geometry}
\usepackage[singlespacing]{setspace}
\usepackage[bottom]{footmisc}
\usepackage{indentfirst}
\usepackage{endnotes}
\usepackage{graphicx}%
\usepackage{rotating}

\newtheorem{theorem}{Theorem}
\newtheorem{lemma}{Lemma}
\newtheorem{corollary}[theorem]{Corollary}

\newtheorem{example}{Example}
\newtheorem{assumption}{Assumption}
\newtheorem{definition}{Definition}
\newenvironment{proof}[1][Proof]{\noindent\textbf{#1.} }{\ \rule{0.5em}{0.5em}}
\begin{document}

\title{\LARGE\bf On Non-Bayesian Social Learning}
\author{Pooya Molavi and Ali Jadbabaie}\affil{Department of Electrical and Systems Engineering, University of Pennsylvania,\\ Philadelphia, PA 19104, USA}\date{}
\maketitle

\sloppy

\singlespacing

\begin{abstract}
We study a model of information aggregation and social learning recently proposed by Jadbabaie, Sandroni, and Tahbaz-Salehi, in which individual agents try to learn a correct state of the world by iteratively updating their beliefs using private observations and beliefs of their neighbors. No individual agent's private signal might be informative enough to reveal the unknown state. As a result, agents share their beliefs with others in their social neighborhood to learn from each other. At every time step each agent receives a private signal, and computes a Bayesian posterior as an intermediate belief. The intermediate belief is then averaged with the belief of neighbors to form the individual's belief at next time step. We find a set of minimal sufficient conditions under which the agents will learn the unknown state and reach consensus on their beliefs without any assumption on the private signal structure. The key enabler is a result that shows that using this update, agents will eventually forecast the indefinite future correctly.
\end{abstract}


\section{Introduction}
\label{sec:intro}
We discuss a model of how individuals might use relevant information available to them to form opinions about different social, economic, and political issues. Oftentimes, the ``best'' action an agent can take is not obvious and depends on certain unknown parameters that determine the agent's utility function. Consider the example of an institution making an investment decision. The expected utility of different investment options depend, among other things, on the expected policies of the next administration, which in turn depend on how likely different political candidates are to win the election. One can think of the candidate who is the most likely to win as an unknown parameter in the investment decision of institutions. We assume that the value of the parameter is an unknown state of the world that is exogenous to our model and is determined at time zero by nature. 

Agents make relevant observations about the unknown state. Frequently, not all the agents make the same observations, and not all the observations are equally informative. Lack of access to all the relevant information is a motivation for individuals to share their opinions with others in order to learn from their personal experiences. The importance of social interactions on opinion dynamics have been documented in several examples such as consumer choice, and diffusion of medical innovations and agricultural practices~\cite{Goyal10}. 

We study the scenario in which each individual has two sources of information: his personal observations, and those of his neighbors in a social network (e.g. friends, neighbors, colleagues, etc.). However, agents might not have direct access to personal experiences of their neighbors. Instead, we assume that they can only observe their neighbors' beliefs, i.e., their subjective probabilities of different feasible realizations of the unknown state of the world. An alternative equally valid interpretation is that agents play a particular repeated game of imperfect information where each agent can only observe the actions of his neighbors, and the actions completely reveal the beliefs of acting agents.



To study the effect of social networks on learning, we analyze the model in \cite{Tahbaz_Sandroni_Jadbabaie09} in which agents use an update other than Bayes' rule to incorporate the new information available to them. Agents repeatedly interact with their neighbors and use a simple rule to update their beliefs. Each agent first forms the Bayesian posterior given his observed private signal, as an intermediate step. He then updates his belief to the convex combination of his Bayesian posterior and the beliefs of his neighbors. 

There are two major motivations for considering non-Bayesian protocols: The positive point of view comes from the choice theoretic studies showing existence of non-Bayesian opinion dynamics among people (see \cite{Kahneman_Tversky74, Camerer95,Rabin98} for instance). The second motivation comes from a normative point of view. Bayesian inference in social networks can be, except for certain simple scenarios, computationally complicated to carry out. Part of the complications is because there is no reason to believe that agents know the source of their neighbors' information. Rather, they have to infer it to be able to form an unbiased belief about the true state of the world. The complexities of Bayesian updating limit its applicability in practice. 

In \cite{Tahbaz_Sandroni_Jadbabaie09}, the authors show that, under some assumptions, this update eventually leads to social learning, even in finite networks: They show that agents can eventually forecast the \emph{immediate future} correctly. Furthermore, they will eventually learn the unknown state, if for each agent there exist a signal that is the more probable under the true state of the world than any other state. This assumption ``turns the deck'' in favor of learning by assuming that agents are infinitely often notified, indirectly, of the true state of the world.

We argue through a simple example that the assumption of having a ``revealing'' signal is very limiting. We show that agents will learn the state of the world under the much weaker assumption that they can distinguish the state collectively. Signals need not be independent among agents at the same time period. Instead, we require the signal structure to be such that the state is identifiable given the marginals of the likelihood function. We first prove that not only agents will forecast the immediate future correctly, but also they eventually will learn to forecast the indefinite future. We also show that there exist a signal sequence which is informative enough to let agents identify the true state of the world, even if no revealing signal exists. The results signify that even when none of the agents have enough information to learn the true state of the world, and in spite of individual signals not being revealing enough, social interaction can aggregate pieces of information available to the individuals such that each and everyone of them can distinguish the true state of the world. Furthermore, this paper suggests a distributed, computationally tractable algorithm for learning in networks that can be applied to real world problems such as parameter estimation in sensor networks (cf. \cite{Xiao_Boyd_Lall05} and the references therein).

Recently there has been a rapidly growing interest in models of learning in networks. For a survey of different models of non-strategic and strategic social learning see \cite{Goyal10}. Since in the presence of social networks fully Bayesian inference faces tractability problems, different types of simplifications have been proposed. The first group of models assumes that agents interact sequentially. Examples include models in \cite{Banerjee92, Bikhchandani_Hirshleifer_Welch92, Smith_Sorensen00, Banerjee_Fudenberg04, Lobel08}. In such models each agent, having observed the actions of some of the agents acting before, takes an action. Each agent makes only one decision, and cannot reverse or change his choice. The other group of models suggests non-Bayesian rules of thumb for belief update. Examples include \cite{DeGroot74, Ellison_Fudenberg93, Ellison_Fudenberg95, Bala_Goyal98, DeMarzo03, Golub_Jackson10, Parandeh10, Tahbaz_Sandroni_Jadbabaie09}.


\section{The Model}
\label{sec:model}
The social learning model we consider, was first proposed in \cite{Tahbaz_Sandroni_Jadbabaie09}. Time is discrete and there is a finite number of agents, signals, and states of the world. 

Let $\Theta$ be the finite set of possible states of the world, and $\theta^*\in\Theta$ be the true state that is determined at time zero by nature, and is unchanged thereafter. Let $\mathcal{N}=\{1,2,\dots,n\}$ be the set of agents. At time $t\geq 0$ each agent $i$ has a belief about the true state, denoted by $\mu_{i,t}(\theta)$, which is a probability distribution over $\Theta$. 

At each time period $t\geq 1$ each agent $i$ observes a private random signal $\omega_{i,t}\in S_i$ where $S_i=\{s_i^1,s_i^2,\dots,s_i^{M_i}\}$ is the set of possible signals for agent $i$. Conditional on $\theta\in\Theta$ being the state of the world, the observation profile $\omega_t=(\omega_{1,t},\omega_{2,t},\dots,\omega_{n,t})$ is generated according to the likelihood function $\ell(\cdot\vert\,\theta)$ with $\ell_i(\cdot\vert\,\theta)$ as its $i$th marginal. Let $\mathbb{P}^\theta=\ell(\cdot\vert\,\theta)^\mathbb{N}$ be the product measure that determines the realization of signals conditioned on $\theta$ being the state of the world. This definition allows for signals to be correlated among agents at the same time period, but makes them independent over time. Without loss of generality we assume that $\ell_i(s_i\,\vert\theta^\ast)>0$ for all $s_i\in S_i$; $S_i$ is only the set of signals that are realized with positive probability conditioned on the true state of the world being $\theta^*$.

Let $\bar{\Theta}_i=\{\theta\in\Theta:\ell_i(s_i\vert\,\theta)=\ell_i(s_i\vert\,\theta^\ast) \text{ for all } s_i\in S_i\}$ be the set of states that are observationally equivalent to the true state $\theta^\ast$ from the point of view of agent $i$, and let $\bar{\Theta}=\bar{\Theta}_1\cap\dots\cap\bar{\Theta}_n$ be the set of states that are observationally equivalent to the true state of the world from the point of view of all agents. 

$(\Omega,\mathcal{F},\mathbb{P})$ is the probability triple, where $\Omega=(\prod_{i=1}^n {S_i})^\mathbb{N}$, $\mathcal{F}$ is the smallest $\sigma$-field that makes all $\omega_{i,t}$ measurable, and $\mathbb{P}=\mathbb{P}^{\theta^*}$ is the probability distribution determining the realization of signals. $\mathbb{N}$ stands for the set of natural numbers. We use $\omega\in\Omega$ to denote the infinite signal sequence $(\omega_1,\omega_2,\dots)$, and $\mathbb{E}$ to denote the expected value operator with respect to the probability measure $\mathbb{P}$. Let $\mathcal{F}_{i,t}=\sigma(\omega_{i,1},\omega_{i,2},\dots,\omega_{i,t})$ be the filtration generated by the observations of agent $i$ up to time $t$, and let $\mathcal{F}_t$ be the $\sigma$-field generated by the union of all $\mathcal{F}_{i,t}$ for $1\leq i\leq n$. 

We say that the adapted random variables $X_t$ and $Y_t$ are asymptotically $\mathbb{P}$-almost surely equal, denoted by $X_t\overset{a.a.s.}{=}Y_t$, if there exist $\tilde{\Omega}\subseteq\Omega$ such that $\mathbb{P}(\tilde{\Omega})=1$, and for all $\omega\in\tilde{\Omega}$ and all $\epsilon>0$, there exist $T(\omega,\epsilon)$ such that for all $t_1,t_2> T(\omega,\epsilon)$, 
\begin{equation*}
\lvert X_{t_1}-Y_{t_2}\rvert<\epsilon.
\end{equation*}
It is an easy exercise to show that if $X_t\overset{a.a.s.}{=}Y_t$ and $Z_t\overset{a.a.s.}{=}W_t$, then $X_t \pm Z_t\overset{a.a.s.}{=}Y_t \pm W_t$, and $X_t Z_t \overset{a.a.s.}{=} Y_t W_t$.

The interactions between the agents are captured by a directed graph $G=(\mathcal{N},E)$. Let $\mathcal{N}_i = \{j\in\mathcal{N}: (j,i)\in E\}$ be the set of neighbors of agent $i$. It is assumed that agent $i$ can observe the belief of agent $j$ if there exist a directed edge from $i$ to $j$, that is $(i,j)\in E$. A graph is called \emph{strongly connected} if there exist a directed path from any vertex to any other one. 

Each agent $i$ starts with the initial belief $\mu_{i,0}(\theta)$ that $\theta$ is the true state of the world. At the end of period $t$, each agent observes the beliefs of his neighbors. At the beginning of the next period, agent $i$ receives the private signal $\omega_{i,t+1}$, and then uses the following rule to update his belief:
\begin{equation}\label{eq:dynamic}
\mu_{i,t+1}(\theta)=a_{ii}\mu_{i,t}(\theta)\frac{\ell_i(\omega_{i,t+1}\vert\,\theta)}{m_{i,t}(\omega_{i,t+1})}+\sum_{j\in\mathcal{N}_i}a_{ij}\mu_{j,t}(\theta),
\end{equation}
where $m_{i,t}(s_i)$ is defined for any $s_i\in S_i$ as
\begin{equation*}\label{def_m}
m_{i,t}(s_i) = \sum_{\theta\in\Theta} \ell_i(s_i\vert\,\theta)\mu_{i,t}(\theta).
\end{equation*}
In the update in equation (\ref{eq:dynamic}) each agent updates his belief to a convex combination of his own Bayesian posterior, given only his private signal and neglecting the social network, and his neighbors' previous period beliefs. $a_{ij}$ is the weight agent $i$ assigns to the opinion of agent $j$, and $a_{ii}$, called the \emph{self-reliance} of agent $i$, is the weight he assigns to his Bayesian posterior conditional on his private signal. We assume that $a_{ij}\in\mathbb{R}^{+}$ and $\sum_{j\in\mathcal{N}_i\cup\{i\}}a_{ij}=1$ for the beliefs to remain a probability distribution over $\Theta$ after the update. 

It is sometimes more convenient to use vector notations. We use $A$ to denote the $n\times n$ matrix whose $ij$ element is $a_{ij}$, and use $\mu_t(\theta)$ to denote the $n$ dimensional column vector whose $i$th element is $\mu_{i,t}(\theta)$. 

When there is no arrival of new information, this update becomes the same as DeGroot's naive learning model. Likewise, when $a_{ij}=0$ for all $j\neq i$, the model is the same as the Bayesian learning model with no network structure.

For all $t>0$, $\mu_{i,t}(\theta)$ and $m_{i,t}(s_i)$ are random functions adapted to $\mathcal{F}_{t}$, the former on $\Theta$ and the latter on $S_i$. $m_{i,t}(s_i)$ is the probability that agent $i$ assigns, at time $t$, to signal $s_i$ being observed in the next time step, hence, it is called agent $i$'s \emph{one step forecast}. We can extend this notion to define the \emph{$k$-step forecast} $m_{i,t}(s_{i,1},s_{i,2},\dots,s_{i,k})$ as the forecast at time $t$ of agent $i$ that the signal sequence $(s_{i,1},s_{i,2},\dots,s_{i,k})\in (S_i)^k$ will be realized in the next $k$ time steps
\begin{equation*}
\quad\displaystyle m_{i,t}(s_{i,1},\dots,s_{i,k})=\sum_{\theta\in\Theta}{ \ell_i(s_{i,1},\dots,s_{i,k}\vert\,\theta)\mu_{i,t}(\theta)},
\end{equation*}
where $\ell_i(s_{i,1},s_{i,2},\dots,s_{i,k}\vert\,\theta)$ is the shorthand for $\ell_i(s_{i,1}\vert\,\theta)$ $\ell_i(s_{i,2}\vert\,\theta)\dots \ell_i(s_{i,k}\vert\,\theta)$.

The update in equation (\ref{eq:dynamic}) is local in the sense that each agent only needs the beliefs of his immediate neighbors to compute it. 

\section{Asymptotic Learning}

In this section we find a set of sufficient conditions for learning when agents use equation (\ref{eq:dynamic}) to update their beliefs. First we have to define what we mean by learning. For a discussion of different notions of learning and how they are related see \cite{Blackwell_Dubins62}. The first one we examine here is the notion of \emph{weak merging of opinions}.
\begin{definition}
Agent $i$ learns the likelihood function on a sample path $\omega$, in the sense of weak merging of opinions, if along that path 
\begin{equation*}
m_{i,t}(\cdot)\to \ell_i(\cdot\vert\,\theta^\ast)\quad\text{as}\quad t\to\infty.
\end{equation*}
\end{definition}

When an agent learns the likelihood function in this sense, he will know the probability distribution according to which signals are generated. In other words, the agent can forecast the \emph{immediate future} correctly, as if he knows the likelihood function.

In \cite{Tahbaz_Sandroni_Jadbabaie09}, the authors show that if agents use the update in equation (\ref{eq:dynamic}), learning in the sense of weak merging of opinions occurs under the following assumptions:
\begin{assumption}\label{asmp1}
\mbox{}
\begin{itemize}
\item[(a)] The social network is strongly connected.
\item[(b)] All agents have strictly positive self-reliances.
\item[(c)] There exists an agent with positive prior belief on the true parameter $\theta^\ast$.
\end{itemize}
\end{assumption}

Assumption \ref{asmp1}.(a) allows for information to flow from any agent to any other one. Assumption \ref{asmp1}.(b) is to prevent agents from disregarding their personal experiences. Assumption \ref{asmp1}.(c) is what is known as a ``grain of truth'' in agent's prior belief. In the discrete setting, this is equivalent to absolute continuity of the initial forecasts with respect to the likelihood function, which is often a necessary condition for learning. The following theorem shows that weak merging of opinions occurs when Assumptions \ref{asmp1}(a)-(c) hold.

\begin{theorem}[Jadbabaie, Sandroni, \& Tahbaz-Salehi~\cite{Tahbaz_Sandroni_Jadbabaie09}]\label{m_to_l}
If Assumption \ref{asmp1} holds, then
\begin{equation*}
m_{i,t}(\cdot)\to \ell_i(\cdot\vert\,\theta^\ast)\quad\text{as}\quad t\to\infty,
\end{equation*}
with $\mathbb{P}$-probability one.
\end{theorem}

It can also easily be shown through counterexamples that none of the assumptions of the theorem can be relaxed. This theorem corresponds to weak merging of opinions. The much stronger notion of learning is asymptotic learning of the true state of the world.
\begin{definition}
Agent $i$ asymptotically learns the true parameter $\theta^\ast$ on a sample path $\omega$, if along that path 
\begin{equation*}
\mu_{i,t}(\theta^\ast)\to 1\quad\text{as}\quad t\to\infty.
\end{equation*}
\end{definition}

To prove asymptotic learning authors add the assumption that for any agent $i$, there exists a signal $\hat{s}_i\in S_i$ and a positive number $\delta_i$ such that

\begin{equation}\label{eq:star} \frac{\ell_i(\hat{s}_i\vert\,\theta)}{\ell_i(\hat{s}_i\vert\,\theta^\ast)}\leq\delta_i<1\quad\forall\theta\notin\bar{\Theta}_i.
\end{equation}

This assumption asks for existence of a signal that is more likely conditioned on $\theta^\ast$ being the true state of the world rather than conditioned on any other state in $\bar{\Theta}_i$ being the true state of the world. Under this assumption and provided that the conditions of Theorem \ref{m_to_l} hold, the authors prove that all agents asymptotically learn the true parameter $\theta^\ast$ with $\mathbb{P}$-probability one. The condition in (\ref{eq:star}) guarantees that there exist a revealing signal that is observed infinitely often. The following example shows that this is a weak assumption.
\begin{example}
\label{eg}
Consider a strongly connected graph on two agents. Assume that $\Theta=\{\theta^\ast, \theta_1,\theta_2\}$, and $S_1=S_2=\{H,T\}$. Also assume that the private signals of agents are independent and are generated according to the same probability distribution $\ell(s\vert\,\theta)$ which is given by Table \ref{tab}.
\begin{table}[!t]
\centering
\caption{The likelihood function of Example \ref{eg}}
\label{tab}
\renewcommand{\arraystretch}{1.5}
\begin{tabular}{l|c|c|}
\multicolumn{1}{r}{}
 &  \multicolumn{1}{c}{$H$}
 & \multicolumn{1}{c}{$T$} \\
\cline{2-3}
$\theta^\ast$ & $1/3$ & $2/3$ \\
\cline{2-3}
$\theta_1$ & $1/4$ & $3/4$ \\
\cline{2-3}
$\theta_2$ & $3/5$ & $2/5$ \\
\cline{2-3}
\end{tabular}
\end{table}
In this example there is no signal that satisfies condition (\ref{eq:star}). But the signal sequence $(H,T,T)$ is most likely under $\theta^\ast$ rather than any other $\theta\in\Theta$. Furthermore, $(H,T,T)$ is the shortest such signal sequence.
\end{example}

In this paper we show that to prove asymptotic learning no assumption should be made other than the \emph{distinguishability} of the true state of the world $\theta^\ast$.
\begin{assumption}[Distinguishability]\label{disting}
There is no $\theta\in\Theta$ that is observationally equivalent to $\theta^\ast$ from the point of view of all agents, that is
\begin{equation*}
\bar{\Theta}=\bar{\Theta}_1\cap\cdots\cap\bar{\Theta}_n=\{\theta^\ast\}.
\end{equation*}
\end{assumption}

This is obviously a necessary condition for the agents to learn the true state of the world. In Theorem \ref{thm_main} which is the main result of this section we show that it is also sufficient. To that end, we first show in Theorem \ref{m_to_l_k} that the correct forecasts of the agents can be extended into the future. To prove these results we first need to present a few preliminary results. The following theorem is a variation of Borel-Cantelli lemma which can be found, among many other places, as Theorem 5.3.2 in \cite{Durrett10}.
\begin{theorem}\label{thm_B-C}
Let $\mathcal{F}_t$, $t\geq0$ be a filtration with $\mathcal{F}_0=\{\emptyset,\Omega\}$ and $A_t$, $t\geq1$ a sequence of events with $A_t\in\mathcal{F}_t$. Then
\begin{equation*}
\{A_t \text{ infinitely often}\}=\Bigg\{\sum_{t=1}^{\infty}\mathbb{P}(A_t\vert\mathcal{F}_{t-1})=\infty\Bigg\}.
\end{equation*}
\end{theorem}

The next theorem we need is the bounded convergence theorem for conditional expectations. For a proof see, for instance, page 263 of \cite{Durrett10}.
\begin{theorem}\label{dom_conv}
Suppose $Y_t\to Y$, $\mathbb{P}$-almost surely, and $\lvert Y_t\rvert\leq M$ for all $t$ where $M$ is a constant. If $\mathcal{F}_t\uparrow\mathcal{F}_\infty$ then with $\mathbb{P}$-probability one
\begin{equation*}
\mathbb{E}(Y_t\vert\,\mathcal{F}_t)\to\mathbb{E}(Y\vert\,\mathcal{F}_\infty).
\end{equation*}
\end{theorem}

The next two lemmas are technical lemmas which will be used in proof of Theorem \ref{m_to_l_k}.
\begin{lemma}\label{E_lm}
If Assumption \ref{asmp1} holds then
\begin{equation*}
\mathbb{E}(\frac{\ell_i(\omega_{i,t+1}\vert\,\theta)}{m_{i,t}(\omega_{i,t+1})}\vert\,\mathcal{F}_t)\to1\quad\text{as}\quad t\to\infty,
\end{equation*}
with $\mathbb{P}$-probability one for all $\theta\in \Theta$.
\end{lemma}

\begin{proof}
\begin{IEEEeqnarray*}{rCl}
\mathbb{E}(\frac{\ell_i(\omega_{i,t+1}\vert\,\theta)}{m_{i,t}(\omega_{i,t+1})}\vert\,\mathcal{F}_t) & = & \sum_{s_i\in S_i}\frac{\ell_i(s_i\vert\,\theta)}{m_{i,t}(s_i)}\ell_i(s_i\vert\,\theta^\ast)\\ & \leq &\max_{s_i\in S_i}\frac{\ell_i(s_i\vert\,\theta^\ast)}{m_{i,t}(s_i)}\sum_{s_i\in S_i}\ell_i(s_i\vert\,\theta)\\& = &\max_{s_i\in S_i}\frac{\ell_i(s_i\vert\,\theta^\ast)}{m_{i,t}(s_i)}.
\end{IEEEeqnarray*}
By a similar argument
\begin{equation*}
\mathbb{E}(\frac{\ell_i(\omega_{i,t+1}\vert\,\theta)}{m_{i,t}(\omega_{i,t+1})}\vert\,\mathcal{F}_t)\geq\min_{s_i\in S_i}\frac{\ell_i(s_i\vert\,\theta^\ast)}{m_{i,t}(s_i)}.
\end{equation*}
The conditional expectation is sandwiched between two quantities, both of which go to one $\mathbb{P}$-almost surely by Theorem \ref{m_to_l}, and so does it.
\end{proof}

\begin{lemma}\label{m_ij} If Assumption \ref{asmp1} holds then
\begin{equation*}
\mathbb{E}^*(\mu_{t+1}(\theta)\vert\,\mathcal{F}_t)\overset{a.a.s.}{=}A \mu_{t}.
\end{equation*}
\end{lemma}

\begin{proof}
We take expectations of both sides of equation (\ref{eq:dynamic}) conditioned on $\mathcal{F}_t$. Since $\mu_{j,t}(\theta)$ is $\mathcal{F}_t$ measurable for all $1\leq j\leq n$,
\begin{equation*}
\mathbb{E}(\mu_{i,t+1}(\theta)\vert\,\mathcal{F}_t)  =   a_{ii}\mu_{i,t}(\theta)\mathbb{E}(\frac{\ell_i(\omega_{i,t+1}\vert\,\theta)}{m_{i,t}(\omega_{i,t+1})}\vert\,\mathcal{F}_t) + \sum_{j\in\mathcal{N}_i}a_{ij}\mu_{j,t}(\theta).
\end{equation*}
Taking the limit of the above equation as $t\to\infty$ and using Lemma \ref{E_lm}, 
\begin{equation*}
\mathbb{E}(\mu_{i,t+1}(\theta)\vert\,\mathcal{F}_t)\overset{a.a.s.}{=}a_{ii}\mu_{i,t}(\theta)+\sum_{j\in\mathcal{N}_i}a_{ij}\mu_{j,t}(\theta),
\end{equation*}
which in vector form can be written as 
\begin{equation*}
\mathbb{E}(\mu_{t+1}(\theta)\vert\,\mathcal{F}_t)\overset{a.a.s.}{=}A \mu_{t}.
\end{equation*}
\end{proof}

The next theorem shows that not only the agents eventually forecast the next step correctly, as shown in Theorem \ref{m_to_l}, but also they do so for the next $k$ steps for any \emph{finite} $k$.
\begin{theorem}\label{m_to_l_k}
If Assumption \ref{asmp1} holds, then
\begin{equation*}
\quad\displaystyle m_{i,t}(s_{i,1},\dots,s_{i,k})\to \ell_i(s_{i,1},\dots,s_{i,k}\vert\,\theta^\ast)\quad\text{as}\quad t\to\infty,
\end{equation*}
with $\mathbb{P}$-probability one for all $s_{i,1},s_{i,2},\dots,s_{i,k}\in S_i$.
\end{theorem}

\begin{proof}
To simplify notation we drop the subscript $i$ from $s_{i,1},s_{i,2},\dots,s_{i,k}$ whenever there is no risk of confusion. We use induction on $k$. For $k=1$ the result is proved in Theorem \ref{m_to_l}. Multiplying both sides of equation (\ref{eq:dynamic}) by $m_{i,t}(\omega_{i,t+1})\ell_i(s_2,\dots,s_k\vert\,\theta)$ and summing over $\theta\in\Theta$,
\begin{IEEEeqnarray*}{rCl}
m_{i,t}(\omega_{i,t+1})\sum_{\theta\in\Theta}\ell_i(s_2,\dots,s_k\vert\,\theta)\mu_{i,t+1}(\theta) & = & a_{ii}\sum_{\theta\in\Theta}\ell_i(\omega_{i,t+1},s_2,\dots,s_k\vert\,\theta)\mu_{i,t}(\theta)\\ && +\:m_{i,t}(\omega_{i,t+1})\sum_{\theta\in\Theta}\ell_i(s_2,\dots,s_k\vert\,\theta)\sum_{j\in\mathcal{N}_i}a_{ij}\mu_{j,t}(\theta).
\end{IEEEeqnarray*}
Thus,
\begin{IEEEeqnarray*}{rCl}
m_{i,t}(\omega_{i,t+1})m_{i,t+1}(s_2,\dots,s_k) & = & a_{ii}m_{i,t}(\omega_{i,t+1},s_2,\dots,s_k)\\ && +\:m_{i,t}(\omega_{i,t+1})\sum_{\theta\in\Theta} \ell_i(s_2,\dots,s_k\vert\,\theta)\sum_{j\in\mathcal{N}_i}a_{ij}\mu_{j,t}(\theta).
\end{IEEEeqnarray*}
By Lemma \ref{m_ij},
\begin{IEEEeqnarray*}{rCl}
m_{i,t}(\omega_{i,t+1})m_{i,t+1}(s_2,\dots,s_k) & \overset{a.a.s.}{=} & a_{ii}m_{i,t}(\omega_{i,t+1},s_2,\dots,s_k)\\ && +\:m_{i,t}(\omega_{i,t+1})\sum_{\theta\in\Theta}\ell_i(s_2,\dots,s_k\vert\,\theta)\mathbb{E}(\mu_{i,t+1}(\theta)\rvert\mathcal{F}_t)\\&&-\:a_{ii}m_{i,t}(\omega_{i,t+1})\sum_{\theta\in\Theta} \mu_{i,t}(\theta)\ell_i(s_2,\dots,s_k\vert\,\theta).
\end{IEEEeqnarray*}
Since all the terms are positive and $\ell_i(s_2,\dots,s_k\vert\,\theta)$ is a constant, using Fubini's theorem \cite{Durrett10} we can change the order of sum and expectation to get
\begin{IEEEeqnarray*}{rCl}
m_{i,t}(\omega_{i,t+1})m_{i,t+1}(s_2,\dots,s_k) & \overset{a.a.s.}{=} & a_{ii}m_{i,t}(\omega_{i,t+1},s_2,\dots,s_k)\\&&+\:m_{i,t}(\omega_{i,t+1})\mathbb{E}(m_{i,t+1}(s_2,\dots,s_k)\vert\,\mathcal{F}_t)\\&&-\:a_{ii}m_{i,t}(\omega_{i,t+1})m_{i,t}(s_2,\dots,s_k).
\end{IEEEeqnarray*}
By induction hypothesis $m_{i,t+1}(s_2,\dots,s_k)$ converges $\mathbb{P}$-almost surely to $\ell_i(s_2,\dots,s_k\vert\,\theta^\ast)$. Since $m_{i,t+1}$ is a probability measure it is bounded for all $t$. Also note that $\mathcal{F}_t\uparrow\mathcal{F}$. Hence, we can use Theorem \ref{dom_conv} to conclude that $\mathbb{E}(m_{i,t+1}(s_2,\dots,s_k)\vert\,\mathcal{F}_t)$ converges $\mathbb{P}$-almost surely to $\mathbb{E}(\ell_i(s_2,\dots,s_k\vert\,\theta^\ast)\vert\,\mathcal{F})$ which is just $\ell_i(s_2,\dots,s_k\vert\,\theta^\ast)$. Therefore, 
\begin{IEEEeqnarray*}{rCl}
m_{i,t}(\omega_{i,t+1})\ell_i(s_2,\dots,s_k\vert\,\theta^\ast) & \overset{a.a.s.}{=} & a_{ii}m_{i,t}(\omega_{i,t+1},s_2,\dots,s_k)\\ && +\:m_{i,t}(\omega_{i,t+1})\ell_i(s_2,\dots,s_k\vert\,\theta^\ast)\\ && -\:a_{ii}m_{i,t}(\omega_{i,t+1})\ell_i(s_2,\dots,s_k\vert\,\theta^\ast).
\end{IEEEeqnarray*}
We can now solve for $m_{i,t}(\omega_{i,t+1},s_2,\dots,s_k)$ to get 
\begin{equation*}
m_{i,t}(\omega_{i,t+1},s_2,\dots,s_k)\overset{a.a.s.}{=}m_{i,t}(\omega_{i,t+1})m_{i,t}(s_2,\dots,s_k).
\end{equation*}
By the induction hypothesis, $m_{i,t}(s_2,\dots,s_k)$ converges on a set of $\mathbb{P}$-probability one to $\ell_i(s_2,$$\dots,s_k\vert\,\theta^\ast)$. Moreover, $m_{i,t}(\omega_{i,t+1})$ is asymptotically $\mathbb{P}$-almost surely equal to $\ell_i(\omega_{i,t+1}\vert\,\theta^\ast)$. Therefore, $m_{i,t}(\omega_{i,t+1},s_2,$$\dots,s_k)$ is asymptotically $\mathbb{P}$-almost surely equal to $\ell_i(\omega_{i,t+1},s_2,\dots,s_k\vert\,\theta^\ast)$. To get the result for any arbitrary $s_1,s_2,\dots,s_k\in S_i$, we need to use Theorem \ref{thm_B-C}. For any $\epsilon>0$, let 
\begin{equation*}
A_t=\Bigg\{\omega:\Big\lvert\frac{m_{i,t-1}(\omega_{i,t},s_2,\dots,s_k)}{\ell_i(\omega_{i,t},s_2,\dots,s_k\vert\,\theta^\ast)}-1\Big\rvert>\epsilon\Bigg\},
\end{equation*}
and for any $s_1\in S_i$, let
\begin{equation*}
B_t(s_1)=\Bigg\{\omega:\Big\lvert\frac{m_{i,t-1}(s_1,s_2,\dots,s_k)}{\ell_i(s_1,s_2,\dots,s_k\vert\,\theta^\ast)}-1\Big\rvert>\epsilon\Bigg\}.
\end{equation*}
Clearly $A_t$ is measurable with respect to $\mathcal{F}_t$ for all $t\geq 1$. Since asymptotically $\mathbb{P}$-almost surely $m_t(\omega_{i,t+1},s_2,\dots,s_k)$ is equal to $\ell_i(\omega_{i,t+1},s_2,\dots,s_k\vert\,\theta^\ast)$, $\mathbb{P}(A_t \text{ i.o.})=0$ for any $\epsilon>0$, where i.o. is the abbreviation for infinitely often. Using Theorem \ref{thm_B-C} for the $A_n$ defined above,
\begin{equation}
\label{eq:sum_P}
\sum_{t=1}^{\infty}\mathbb{P}(A_t\vert\mathcal{F}_{t-1})<\infty,
\end{equation}
with $\mathbb{P}$-probability one. We can write the summand in equation (\ref{eq:sum_P}) as
\begin{equation}
\label{eq:indicator}
\mathbb{P}(A_t\vert\mathcal{F}_{t-1})=\mathbb{E}\Big(\sum_{s_1\in S} 1_{B_t(s_1)} 1_{\{\omega_{i,t}=s_1\}}\vert\mathcal{F}_{t-1}\Big),
\end{equation}
where for any $A\in\mathcal{F}$, $1_A$ is a random variable defined as
\begin{equation*}
1_A(\omega)=\begin{dcases}
   1 & \text{if } \omega\in A, \\
   0       & \text{if } \omega\notin A.
  \end{dcases}
\end{equation*}
$1_{B_t(s_1)}\in\mathcal{F}_{t-1}$ for any $s_1\in S$, and $\omega_{i,t}$ is independent of $\mathcal{F}_{t-1}$. Since all the random variables in equation (\ref{eq:indicator}) are positive, we can switch the order of sum and expected value to get
\begin{IEEEeqnarray*}{rCl}
\mathbb{P}(A_t\vert\mathcal{F}_{t-1}) & = & \sum_{s_1\in S}1_{B_t(s_1)}\mathbb{E} 1_{\{\omega_{i,t}=s_1\}}\\ & = &\sum_{s_1\in S}1_{B_t(s_1)}\ell_i(s_1\,\vert\theta^\ast).
\end{IEEEeqnarray*}
If we substitute the above equation for $\mathbb{P}(A_t\vert\mathcal{F}_{t-1})$ in equation (\ref{eq:sum_P}) and use the fact that all the terms are positive to switch the order of the sums, we get that
\begin{equation*}
\sum_{s_1\in S}\ell_i(s_1\,\vert\theta^\ast)\sum_{t=1}^{\infty}1_{B_t(s_1)}<\infty,
\end{equation*}
with $\mathbb{P}$-probability. Therefore, $\mathbb{P}(B_t(s_1) \text{ i.o.})=0$ for any $s_1\in S_i$, and $m_{i,t}(s_1,s_2,\dots,s_k)$ converges $\mathbb{P}$-almost surely to $\ell_i(s_1,s_2,\dots,s_k\vert\,\theta^\ast)$ for any $s_1,s_2,\dots,s_k\in S_i$.
\end{proof}

This theorem shows that, eventually, an agent's forecast at time $t$ of a time $t+k$ event is not different than his forecast at time $t+k-1$ of the same event. That is, when agents already have accurate forecasts of immediate future, they cannot improve their forecasts of indefinite future by observing more signals.

The following theorem uses finiteness of $S_i$ to show that for any agent $i$ there exist a long enough signal sequence that is more probable under $\theta^\ast$ than any state $\theta\notin\bar{\Theta}_i$.
\begin{theorem}\label{k_step_st}
For any agent $i$, there exists a finite number $\hat{k}_i$ and signals $\hat{s}_{i,1},\hat{s}_{i,2},\dots,\hat{s}_{i,\hat{k}_i}$ such that
\begin{equation}
\label{eq:k_step_st} \displaystyle\frac{\ell_i(\hat{s}_{i,1},\hat{s}_{i,2},\dots,\hat{s}_{i,\hat{k}_i}|\,\theta)}{\ell_i(\hat{s}_{i,1},\hat{s}_{i,2},\dots,\hat{s}_{i,\hat{k}_i}|\,\theta^\ast)}\leq\delta_i<1\qquad\forall\theta\notin\bar{\Theta}_i.
\end{equation}
\end{theorem}

\begin{proof}
First assume that $\ell_i(s_i\vert\,\theta^\ast)$ is a rational number for all $s_i\in S_i$. In this case we will prove that we can take $\hat{k}_i$ to be the least common denominator (LCD) of $\{\ell_i(s_i\vert\,\theta^\ast)\}_{s_i\in S_i}$ and $(\hat{s}_{i,1},\hat{s}_{i,2},\dots,\hat{s}_{i,\hat{k}_i})$ to be a sequence of signals in which the number of occurrences of each $s_i\in S_i$ is exactly equal to the numerator of the fractional representation of $\ell_i(s_i\vert\,\theta^\ast)$, when the denominator is equal to $\hat{k}_i$. In other words, we pick the signal sequence in which the frequency of each signal is equal to its probability under $\theta^\ast$. 

Let $\hat{k}_i=\text{LCD}\big(\{\ell_i(s_i\vert\,\theta^\ast)\}_{s_i\in S_i}\big)$ and $k_i^j=\ell(s_i^j\vert\,\theta^\ast)\hat{k}_i$ for $1\leq j\leq M_i$. We prove that $\ell_i(\cdot\vert\,\theta^\ast)$ is the unique probability measure for which the probability of the signal sequence 
\begin{equation*}
(\underbrace{s_i^1,\dots,s_i^1}_{k_i^1\text{times}},\dots,\underbrace{s_i^{M_i},\dots,s_i^{M_i}}_{k_i^{M_i}\text{times}})
\end{equation*}
is maximized. As a result, for this sequence $\displaystyle \ell_i(\cdot|\,\theta)/\ell_i(\cdot|\,\theta^\ast)$ is strictly less than one for any $\theta\notin\bar{\Theta}_i$. 

Let $p_i^j=\mathbb{Q}(s_i^j)$ for $1\leq j\leq M_i$, where $\mathbb{Q}$ is some probability measure on $S_i$. To simplify notation we drop the subscript $i$ whenever there is no risk of confusion. We solve the following concave maximization problem:
\begin{IEEEeqnarray}{rCl}
\label{eq:opt}
\begin{aligned}
\max_{p^1,\dots,p^{M}}&\qquad {(p^1)}^{k^1} {(p^2)}^{k^2}\dots {(p^{M})}^{k^{M}}\\
\text{subject to}&\qquad p^1+p^2+\dots+p^{M}=1,
\end{aligned}
\end{IEEEeqnarray}
where by ${(p^j)}^{k^j}$ we mean $p^j$ to the power of $k^j$. We can use a Lagrange multiplier $\rho$ to incorporate the constraint into the cost function~\cite{Boyd04}. The resulting unconstrained optimization is still concave. Therefore, the optimal solutions can be found using the first order conditions. The resulting set of equalities are
\begin{IEEEeqnarray*}{rCl}
&\frac{k^j}{p^j}-\rho=0\qquad 1\leq j\leq M,\\
&p^1+p^2+\dots+p^m=1,
\end{IEEEeqnarray*}
which have the unique solution 
\begin{equation*}
p^j=\frac{k^j}{\sum_{j=1}^M k^j}\qquad 1\leq j\leq M.
\end{equation*}
By construction, this solution corresponds to the probability measure $\ell_i(\cdot\vert\,\theta^\ast)$. 

So far we have proved the theorem for the rational case. For the case that $\ell_i(s_i\vert\,\theta^\ast)$ is irrational for some $s_i\in S_i$, the proof follows from continuity of the objective of the optimization problem (\ref{eq:opt}) with respect to $(p^1,p^2,\dots,p^j)$, and the fact that rational numbers are dense in reals.
\end{proof}

The condition in (\ref{eq:k_step_st}) is the $k$-step counterpart of the condition in (\ref{eq:star}). We showed that even though a single signal revealing the true state might not exist, but there are signal sequences that will do so. We are now ready to prove the main result of this section:
\begin{theorem}
\label{thm_main}
If Assumption \ref{asmp1} holds, then
\begin{equation*}
\mu_{i,t}(\theta)\to0\quad\text{as}\quad t\to\infty\qquad\forall i\in\mathcal{N},\quad\forall{\theta}\notin\bar{\Theta},
\end{equation*}
with $\mathbb{P}$-probability one for all $i\in\mathcal{N}$.
\end{theorem}

\begin{proof}
First we prove that for any agent $i$, $\mu_{i,t}(\theta)\to0$ as $t\to\infty$ for all $\theta\notin\bar{\Theta}_i$ with $\mathbb{P}$-probability one. Let $\hat{k}_i$ and $(\hat{s}_{i,1},\hat{s}_{i,2},\dots,\hat{s}_{i,\hat{k}_i})$ be a positive integer and a sequence of signals respectively, that satisfy (\ref{eq:k_step_st}). By Theorem \ref{m_to_l_k}, $\,m_{i,t}(s_{i,1},\cdots,s_{i,k})\to \ell_i(s_{i,1},\cdots,s_{i,k}|\,\theta^\ast)$ with $\mathbb{P}$-probability one for any sequence of finite length. We can use this result for $\hat{s}_{i,1},\hat{s}_{i,2},\dots,\hat{s}_{i,\hat{k}_i}$ to conclude that for any $\epsilon>0$, there exist $\tilde{\Omega}_i\subseteq\Omega$ and a random variable $T_i(\omega,\epsilon)$ such that $\mathbb{P}(\tilde{\Omega}_i)=1$, and for any $\omega\in\tilde{\Omega}_i$ and $t>T_i(\omega,\epsilon)$,
\begin{equation*}
\left|\sum_\theta{\mu_{i,t}(\theta)\frac{\ell_i(\hat{s}_{i,1},\dots,\hat{s}_{i,\hat{k}_i}|\,\theta)}{\ell_i(\hat{s}_{i,1},\dots,\hat{s}_{i,\hat{k}_i}|\,\theta^\ast)}}-1\right| < \epsilon,
\end{equation*}
and therefore,
\begin{equation*}
\left|\sum_{\theta\notin\bar{\Theta}_i}{\mu_{i,t}(\theta)\frac{\ell_i(\hat{s}_{i,1},\dots,\hat{s}_{i,\hat{k}_i}|\,\theta)}{\ell_i(\hat{s}_{i,1},\dots,\hat{s}_{i,\hat{k}_i}|\,\theta^\ast)}}+\sum_{\theta\in\bar{\Theta}_i}\mu_{i,t}(\theta)-1\right| < \epsilon.
\end{equation*}
Using the result of Theorem \ref{k_step_st} we can conclude that on any sample path $\omega\in\tilde{\Omega}_i$ and $t>T_i(\omega,\epsilon)$,
\begin{equation*}
0\leq (1-\delta_i)\sum_{\theta\notin\bar{\Theta}_i}{\mu_{i,t}(\theta)} < \epsilon.
\end{equation*}
Since $\epsilon>0$ is arbitrary, this proves that $\mu_{i,t}(\theta)\to0$ as $t\to\infty$ for all $\theta\notin\bar{\Theta}_i$.

Taking limits of equation (\ref{eq:dynamic}) as $t\to\infty$ and using the result proved above shows that $\sum_{j\in\mathcal{N}_i}a_{ij}\mu_{j,t}(\theta)$ converges to zero as $t\to\infty$ and so does $\mu_{j,t}(\theta)$ for all $\theta\notin\bar{\Theta}_i$ and $j\in\mathcal{N}_i$ with $\mathbb{P}$-probability one. Proceeding in the same way and using the strong connectivity assumption, for all $j\in\mathcal{N}$, $\mu_{j,t}(\theta)\to0$ for all $\theta\notin\bar{\Theta}_i$. Thus, with $\mathbb{P}$-probability one
\begin{equation}
\label{eq:nonlin_conv}
\mu_{i,t}(\theta)\to0\quad\text{as}\quad t\to\infty\qquad\forall i\in\mathcal{N},\quad\forall{\theta}\notin\bar{\Theta}.
\end{equation} 
\end{proof}

Since $\mu_{i,t}(\cdot)$ is a probability distribution over $\Theta$ for all $i$ and $t$, we have the following corollary.
\begin{corollary}
\label{cor:main}
If Assumptions \ref{asmp1} and \ref{disting} hold, then
\begin{equation*}
\mu_{i,t}(\theta^\ast)\to 1\quad\text{as}\quad t\to\infty,
\end{equation*}
with $\mathbb{P}$-probability one for all $i\in\mathcal{N}$.
\end{corollary}


\section{Conclusion and Future Direction}
We analyzed a model of social learning where the agents form their beliefs about an unknown state of the world by taking the convex combination of their Bayesian posterior and the beliefs of their neighbors. Agents will learn the true state of the world if, (a) the social network is strongly connected, (b) all agents have strictly positive self-reliances, and (c) there exists an agent with positive prior belief on the true state. Furthermore, none of these assumptions can be relaxed in general. We also argued that there does not need to exist a revealing signal for learning to happen. Rather, there are always long enough signal sequences that are revealing. We also proved that once the immediate future forecasts of the agents become approximately correct, they can be extended to indefinite future with a negligible error.

We also have shown that, under the same assumptions, not only the agents eventually learn the true state of the world, but also they do so exponentially fast with an exponent that depends both on network topology and on signal structure. Therefore, even though network structure had a minimal effect on the possibility of learning through requiring strong connectivity, it influences the rate of learning in a more complicated and subtle way. Moreover, as a corollary of exponential learning, eventually the consensus dynamic becomes dominant. Hence, even if the true state is not distinguishable, agents eventually reach a consensus.

As we mentioned in Section \ref{sec:intro}, the model of social learning where the agents exchange their beliefs with their neighbors can be regarded as a trivial game of imperfect information where the actions of the agents completely reveal their beliefs. We would like to extend the model to nontrivial games where the observed actions only contain limited information about the underlying beliefs. Then, to be able to use the simple update discussed in this paper, agents have to estimate the belief of their neighbors based on the history of their observed actions. We are interested in quantifying how, and to what extent unobservability of the beliefs will change the results of this paper.


%
\section*{Acknowledgments}
The authors would like to thank Alireza Tahbaz-Salehi for careful review of our proofs and many helpful comments and discussions. This research was supported in parts by ONR MURI No. N000140810747, and AFOSR Complex Networks Program.

\bibliographystyle{IEEEtran}
\bibliography{learningNW}

\end{document}